\documentclass[sigplan,screen]{acmart}

\AtBeginDocument{%
  }

\usepackage{mathpartir}
\usepackage{latexsym}
\usepackage{stmaryrd}
\usepackage{xcolor}

\newtheorem{theorem}{Theorem}

\newcommand{\match}[4]{#1 \, @ \, \left\langle #3, #4 \right\rangle \approx #2}
\newcommand{\matchNOFUN}[3]{#1 \, @ \, #3 \approx #2}

\newcommand{\guardPat}[2]{#1 \, ; \texttt{guard} \left(#2\right)}
\newcommand{\backMatch}[3]{#1 \, ; \left(#3 \approx #2\right)}

\newcommand{\scott}[1]{\ensuremath{\llbracket #1 \rrbracket}}

\newcommand{\actDoMatch}[2]{\texttt{match}\left(#1,#2\right)}
\newcommand{\actUnbind}[1]{\texttt{checkName}\left(#1\right)}
\newcommand{\actCheckGuard}[1]{\texttt{guard}\left(#1\right)}
\newcommand{\actBackMatch}[2]{\texttt{matchConstr}\left(#2,#1\right)}

\newcommand{\stkEmp}{[\,]}
\newcommand{\bt}[4]{\left(#2,#3,#4\right) \texttt{::} #1} 

\newcommand{\btNODENOFUN}[2]{\left(#1,#2\right)} 
\newcommand{\btNOFUN}[3]{\btNODENOFUN{#2}{#3} \texttt{::} #1} 

\newcommand{\success}[2]{\texttt{success}\left(#1,#2\right)}
\newcommand{\failure}{\texttt{failure}}
\newcommand{\running}[4]{\texttt{running}\left(#1,#2,#3,#4\right)}

\newcommand{\successNOFUN}[1]{\texttt{success}\left(#1\right)}
\newcommand{\runningNOFUN}[3]{\texttt{running}\left(#1,#2,#3\right)}

\newcommand{\stateTrans}[2]{#1 \mapsto #2}

\newcommand{\contNil}{[\,]}
\newcommand{\contCons}[2]{#1 \texttt{::} #2}

\newcommand{\backtrack}[1]{backtrack\left(#1\right)}

\newif{\ifdraft}
\drafttrue

\newif{\ifextended}
\extendedtrue

\newcommand{\theappendix}{\ifextended Appendix~\ref{app:technical}\else the extended version \cite{extended}\fi}

\newcommand{\dlcb}{DLCB}
\newcommand{\pypm}{PyPM}
\newcommand{\core}{CorePyPM}

\DeclareMathOperator{\arity}{arity}

\newcommand{\rulename}[1]{\textsc{#1}}

\copyrightyear{2025}
\acmYear{2025}
\setcopyright{acmlicensed}\acmConference[CGO '25]{Proceedings of the 23rd
ACM/IEEE International Symposium on Code Generation and
Optimization}{March 1--5, 2025}{Las Vegas, NV, USA}
\acmBooktitle{Proceedings of the 23rd ACM/IEEE International Symposium on
Code Generation and Optimization (CGO '25), March 1--5, 2025, Las Vegas,
NV, USA}
\acmDOI{10.1145/3696443.3708934}
\acmISBN{979-8-4007-1275-3/25/03}

\begin{document}

\title{Pattern Matching in AI Compilers and its Formalization\ifextended~(Extended)\fi}

\author{Joseph W. Cutler}
\authornote{Work carried out while at NVIDIA} 
\email{jwc@seas.upenn.edu}
\affiliation{
    \institution{Univeristy of Pennsylvania}
    \country{United States}
}

\author{Alex Collins}
\email{acollins@nvidia.com}
\affiliation{
    \institution{NVIDIA}
    \country{United States}
}
\author{Bin Fan}
\email{binf@nvidia.com}
\affiliation{
    \institution{NVIDIA}
    \country{United States}
}

\author{Mahesh Ravishankar}
\email{mahesh.ravishankar@amd.com}
\authornotemark[1]
\affiliation{
    \institution{AMD}
    \country{United States}
}

\author{Vinod Grover}
\email{vgrover@nvidia.com}
\authornote{Corresponding Author} 
\affiliation{
    \institution{NVIDIA}
    \country{United States}
}

\begin{abstract}
\pypm{} is a Python-based domain specific language (DSL) for building rewrite-based
optimization passes on machine learning computation graphs.
Users define individual optimizations by writing (a) \emph{patterns} that match
subgraphs of a computation graph and (b) corresponding \emph{rules} which
replace a matched subgraph with an optimized kernel.
\pypm{} is distinguished from the many other DSLs for defining rewriting passes
by its complex and novel pattern language which borrows concepts from logic programming. \pypm{} patterns can be recursive,
nondeterminstic, and can require checking domain-specific constraints such as
the shapes of tensors. The \pypm{} implementation is thus similarly complicated,
consisting of thousands of lines of C++ code.
In this paper, we present our work on building PyPM, as well as formalizing and distilling and this complexity to an
understandable mathematical core.
We have developed a formal core calculus expressing the main operations of the \pypm{} pattern language.
We define both a declarative semantics --- describing which patterns
match which terms --- and an algorithmic semantics --- an idealized version of
the \pypm{} pattern interpreter --- and prove their equivalence. The development
is fully mechanized in the Coq proof assistant.
\end{abstract}

\ccsdesc[500]{Software and its engineering~Compilers}
\ccsdesc[500]{Theory of computation~Equational logic and rewriting}

%
\maketitle

\section{Introduction}
\label{sec:intro}
AI compilers transform machine learning models written in
high-level domain specific languages --- such as PyTorch \cite{pytorch}, TensorFlow \cite{tensorflow}, and Jax \cite{jax} ---
into low-level code to be executed on the GPUs required for high-performance in modern AI applications. 
These compilers are critically important for the AI ecosystem, as they allow AI researchers and engineers to express their code in the familiar style of
tensor algebra and still have them run performantly.
However, the GPU programming interfaces that these compilers must lower to are
complex and evolving rapidly. For AI compiler engineers, this means that
ensuring that their project can emit the highest-performance kernels is a
constantly moving target. For users of AI compilers, it often means that
they cannot utilize the full potential of their hardware without bypassing their
compiler entirely~\cite{flashattention}. 

To underscore the degree of complexity involved in targeting the
highest-performance GPU code, we consider the example of targeting cuBLAS \cite{cublas}, a
BLAS-family \cite{blas} GPU vendor library for GPU-accelerated linear algebra. cuBLAS includes extremely performant hand-tuned kernels (GPU-launched subroutines) that
perform common tensor operations like matrix multiplication. Ideally, an AI compiler
writer could simply lower uses of matrix multiplication to calls to the cuBLAS library.
However, using these
kernels in place of a naive implementation requires checking a great many
constraints first. The cuBLAS kernels work for only a small number of tensor
sizes and shapes, do not support all datatype combinations, and require their
data arguments to be laid out with complex swizzle patterns \cite{swizzleinventor}. A large
bottleneck to compiler engineers targeting these optimized kernels is that it takes
a lot of infrastructure and architectural knowledge to even teach a compiler
to \emph{recognize} when part of a tensor computation written by a user in a high-level language
can be lowered to a use of an optimized kernel. Indeed, with GPU architectures evolving rapidly, we
cannot expect compiler developers to integrate every high-performance kernel
from GPU vendors into their instruction selection and optimization routines.

To this end, we have developed \dlcb{}. \dlcb{} is a standalone GPU compiler backend
for AI compilers. \dlcb{} takes the tensor computation graphs that
are emitted by AI compilers, and transforms them to utilize the
latest GPU tensor instructions or the fastest hand-tuned fused kernels.
We note that we are not the first researchers to notice the problem of fast-moving architectural targets \cite{exo, tvm},
and we are far from the first to develop a bespoke compiler backend to optimize AI model code \cite{tvm, xla}.

The main contribution of this paper is rather a subsystem of \dlcb{} called \pypm{}, which is a new
domain-specific language, embedded in Python, 
designed to help solve the subgraph recognition problem. Programs in \pypm{}
describe both (a) complex \emph{patterns} that are found in tensor computation
graphs, and (b) their replacements, ideally by optimized kernels. \pypm{} is
designed to cleanly express the kind of pattern constraints that are required to
encode the preconditions to use specialized tensor kernels --- things like
tensor shape and element data type constraints, and large nested (and
potentially recursive) chains of linear algebra operations.

As we will see, this expressiveness derives from the fact that \pypm{} supports features beyond those of
most pattern languages. The cost of this expressiveness, however, is an implementation that is challenging
both to write and understand. The part of \dlcb{} that implements a matcher for \pypm{} patterns
is comprised of thousands of lines of C++. This implementation is not based on any well-known algorithm,
and was built incrementally as features were added to PyPM during the langauge design process.
Moreover, in absence of a specification, it is not even clear what it would mean for the code to be ``correct''.

This is an unacceptable state of affairs for a piece of code we intend to run as a critical component of an
AI compiler. To this end, the core of the present project is an effort to develop a formal specification
for what it \emph{means} to match a \pypm{} pattern. We can then distill the large and unwieldy C++ codebase into a
straightforward mathematical algorithm, and prove that it matches the specification.
Even without verifying the C++ code itself, this effort provides us with an added layer of confidence by ensuring that \pypm{}
rests on a sound mathematical foundation.

\subsection*{Pattern Matching, in PyPM and Elsewhere}

Pattern definitions in \pypm{} look like Python methods whose body describes the kinds of subgraphs they match. For example, a pattern definition that returns $$\texttt{MatMul(Trans(x),Trans(y))}$$ matches
operator graphs where two tensor computations, \texttt{x} and \texttt{y} are separately fed into matrix transpose nodes (\texttt{Trans}), and then both passed to a 2-input \texttt{MatMul} node
The body of a pattern definition can have assertions like \texttt{assert x.shape.rank == 2}, which ensures that for this pattern to match, the
subgraph \texttt{x} must produce a rank-2 tensor. Other \pypm{} features include \emph{pattern alternates} --- which behave like disjunctions of patterns, matching if either of the disjuncts match --- \emph{recursive patterns} --- which let users define patterns that match arbitrarily deeply nested subgraphs ---
\emph{nonlinear patterns} --- which let users write patterns like \texttt{MatMul(x,x)} to ensure that the two sides of a matrix multiply are equal ---
and \emph{function patterns} --- which allow users to define patterns like ``any binary operator that produces an \texttt{f32} result''.
Each \pypm{} pattern can have multiple rewrite rules associated with it, each defining a replacement for a
subgraph matched by the pattern. Rules are similarly defined as methods whose body describes the ``right hand side'' of the rewrite. For example, the \texttt{MatMul(Trans(x),Trans(y))}
pattern could have a rewrite rule returning \texttt{Trans(MatMul(y,x))}, replacing the product of transposes by the transpose of the product (in the opposite order).

The \pypm{} langauge is also designed to be highly modular and portable. Rather than implementing a bespoke language toolchain, we have implemented it as an embedded DSL in Python
that can be imported as a library. This ensures that devleopers can use the familiar Python syntax and Pythonic programming style to build complex structures of patterns and rules without having to resort to
a new language. Indeed, pattern and rule definitions can be essentially arbitrary pure Python methods; the translation down into a restricted core calculus of patterns is handled by a symbolic execution-like technique described in Section~\ref{subsec:pypm-impl}.


The concepts of pattern matching, rewriting, instruction selection, and all
sorts of pattern-based compilation have long and rich histories in computer
science --- \pypm{}'s design is inspired by all of these lines of research.
While a full description of related projects can be found in
Section~\ref{sec:related-work}, we touch on a few of \pypm{}'s most important
ancestors here, and describe how it differs from some other points in the
language design space.

The most direct comparison for \pypm{} is any number of languages and systems
for describing and implementing rewriting.  Indeed, interest in (nondestructive)
rewriting systems has recently been rejuvinated by the Egg project \cite{egg, egglog}
and its successs in many domains, including inside tensor and vectorizing
compilers \cite{tensat,diospyros}.
With the more superficial distinctions aside (destructive instead of
nondestructive rewriting), there are two main differences between rewriting
systems like Egg and \pypm{}.  First, the \pypm{} project
is primarily focused on designing and formalizing an expressive, ergonomic, and
portable pattern language for broad use in tensor compilation, as opposed to
building the most efficient general purpose rewriting engine possible. Second,
\pypm{} is not strictly a language for rule-based rewriting --- as previously
discussed, the ability to define complex patterns over computation graphs is
useful for other kinds of optimizaiton as well.

\pypm{} can also be seen as a logic programming language. A pattern
defines a logical formula, where the free variables are pattern variables, and a
satisfying assignment is a match in the graph. Indeed, this connection can be
made explicit by viewing the computation graph as a database of edges between
operator nodes, and \pypm{} patterns as queries \cite{logic-programming}. This view
inspired many of \pypm{}'s features, some of which fall out directly from taking
this analogy seriously. In particular, recursive patterns correspond to
recursive queries, pattern alternates are rules with multiple disjunctive clauses,
and function patterns allow for limited ``second-order'' logic programming.
The main distinction between \pypm{} and a logic language is the evaluation model.
Since \pypm{}'s rewriting is destructive, evaluation is not monotone: creating new facts destroys old facts.

Lastly, many programmers are familiar with pattern matching as a feature in
functional programming languages.  The pattern matching that \pypm{} enables is
superficially similar --- matching a structure in an ADT and generating bindings
for free pattern variables --- but its features go far beyond what would make
sense in a logic language.  In particular, pattern alternates and recursion
would be impossible to implement as efficient code. Some experimental
``pattern-oriented'' programming languages like Egison \cite{egison} do include
features like these, but even then, the use case is quite different.


\subsection*{Contributions}
Our paper has three main contributions.

\begin{enumerate}
    \item First, we contribute the \pypm{} language, a new point in the design space of languages for pattern matching. \pypm{} is designed for use by experts in GPU programming,
    and is expressive enough to describe the kinds of complex patterns of subgraphs that can be replaced by optimized kernels in tensor computations.
    \item We present a formalism for \pypm{}, placing our implementation on a solid mathematical footing. We describe
    a definitional and algorithmic semantics, and proving a soundness theorem. Our work is mechanized in the Coq proof assistant \cite{coq}.
    \item Last, we demonstrate that \pypm{} is works to meet our goals by showing that effective hand-crafted optimizations can be encoded in \pypm{}.
    We also describe other use-cases we have developed for \pypm{} along the way, to show how the ability to write expressive and match patterns is generally 
    useful for AI compiler development.
\end{enumerate}

\section{\pypm{} Language Explanation Via Examples}
\label{sec:examples}

cuBLAS is a library of hand-crafted kernels for accelerated linear algebra. The most important operation in cuBLAS is \texttt{GEMM},
the ``GEneral Matrix Multiply''. \texttt{GEMM} computes $D \leftarrow \sigma \left(\alpha \text{Op}(A)\text{Op}(B) + \beta C\right)$,
where $A$, $B$, and $C$ are matrices, $\alpha$ and $\beta$ are scalars, and $\sigma$ is some pointwise activation function like
\texttt{RELU} or \texttt{GELU} \cite{gelu}. \texttt{GEMM} is a simple example of a high-performance kernel which can be complicated to apply,
in this case because of the sheer number of variants and options available. The library also defines
kernels for a large variety of dimensions, data types (integer and floating point), and activation functions.
Moreover, the $\text{Op}(A)$ in the formula means that the $A$ can be specified to be either $A^T$ or $A$, essentially giving the option to \emph{fuse}
a transpose operation into the call to \texttt{GEMM}. Lastly, not all of these options are compatible, further complicating the question of when using cublas in place of a more basic matrix multiplication is possible.

\begin{figure}
\begin{verbatim}
@op
def MatMul(x,y):
    return 1

@op Trans(x):
    return 1

@op cublasMM_xyT_f32(x,y):
    return 1
    
@op cublasMM_xyT_i8(x,y):
    return 1

@pattern MMxyT(x,y):
    assert x.shape.rank == 2
    assert y.shape.rank == 2
    yt = Trans(y)
    return MatMul(x,yt)

@rule(MMxyT)
def cublasrule(x,y):
    assert (x.eltType == f32 && y.eltType == f32)
        || (x.eltType == i8 && y.eltType == i8)
    if x.eltType == f32 && y.eltType == f32:
        return cublasMM_xyT_f32(x,y)
    elif x.eltType == i8 && y.eltType == i8:
        return cublasMM_xyT_i8(x,y)

\end{verbatim}
\caption{cuBLAS Pattern Example}
\label{fig:mm-to-cublas}
\end{figure}

A \pypm{} program comes in three main parts. The first is a list of declarations of the operators that will appear in patterns.
An operator is declared by defining a method, annotated by an \texttt{@op} decorator, whose name is an alias for the operator. The number of parameters the method has specifies the operator's arity,
and the return value (required to be an integer) defines its output arity.
Operators may have other parameters that are not a part of the dataflow graph: for instance, a convolution operator must specify a \texttt{stride}, but the stride value is not an input or an the computation graph node.
We call these extra parameters ``attributes'', and they can be listed in the operator definition header inside the \texttt{@op} decorator.
Figure~\ref{fig:mm-to-cublas} defines four operators, \texttt{MatMul} (matrix multiply), \texttt{Trans} (transpose), and both \texttt{f32} and \texttt{i8} versions of \texttt{cublasMM\_xyT}, which is an optimized \& fused CUDA kernel for
computing $x \cdot y^T$.

After the operator declarations, a programmer can define patterns.
Pattern definitions are written in \pypm{} as Python methods annotated with a \texttt{@pattern}
decorator. The arguments to the method specify the pattern's free variables, while the return value of the method
is the pattern itself. The pattern looks like a composition of operators (defined above) and the in-scope pattern variables.
For example, Figure~\ref{fig:mm-to-cublas} includes a pattern definition \texttt{MMxyT} which returns \texttt{MatMul(x,Trans(y))}.
This pattern matches against subgraphs of a computation graph that look like \texttt{MatMul(a,Trans(b))}, where \texttt{a} and \texttt{b} are themselves subgraphs.
A successful match then generates the \emph{substitution} $\{\texttt{x} \mapsto \texttt{a}, \texttt{y} \mapsto \texttt{b}\}$, which maps the pattern variables to the concrete subgraphs that make up the match.
Note the use of an intermediate variable \texttt{yt = Trans(y)}. Because patterns can be large and often include a great deal of sharing,
programmers can bind sub-patterns to local names. These local variable names do not get their own bindings in the substitution, they are merely aliases.
The pattern definition \texttt{MMxyT} also demonstrates another feature of \pypm{}, namely \emph{constraints}. This pattern
does not match merely any subgraph that looks like $xy^T$, it requires that both $x$ and $y$ be rank $2$ tensors (i.e. matrices). Constraints can be boolean expressions involving attributes of the pattern variables like \texttt{shape}, or \texttt{eltType},
and are imposed on the pattern using \texttt{assert}. Unlike operators which may have user-defined attributes, all terms (the things that pattern variables bind to) have the same set of tensor-specific attributes including element type, shape, and rank.

Lastly, programmers write rules. Rules are defined in \pypm{} as methods
annotated with an \texttt{@rule(Pat)} annotation, where \texttt{Pat} is the name
of the corresponding pattern. The return value of a rule definition is an expression
that defines the subgraph that will replace the one matched by the corresponding pattern.
Like pattern definitions, rule definitions may contain arbitrary pure Python code, including additional assertions.
Then, when a pattern matches and generates a substitution, \pypm{} runs each of the corresponding rules one by one,
binding the rule's pattern variables to the subgraphs from the substitution.
The first rule whose assertions pass is \emph{fired}, and the subgraph that it defines replaces the root node of the match
in the computation graph. That this rewriting is \emph{destructive}: the matched subgraph is completely replaced by the one defined by the rule.
However, if no rule can apply, then none fires.
Figure~\ref{fig:mm-to-cublas} shows a rule definition called \texttt{cublasrule} for the pattern \texttt{MMxyT}. This rule defines an optimization that replaces subgraphs of the form \texttt{MatMul(x,Trans(y))} with the single-node fused operation
\texttt{cublasMM\_xyT}. It also has an assertion that requires that the elements of the \texttt{x} and \texttt{y} tensors be either both be of type \texttt{i8} or \texttt{f32}. Depending on the particular types, the rule returns different replacements.






\subsection{Pattern Alternates}
The Gaussian Error Linear Unit operator \cite{gelu}, defined as $\text{GELU}(x) = \frac{x}{2} \left(1 + \text{erf}\left(\frac{x}{\sqrt{2}}\right)\right)$, is an activation function used in some popular transformer
models.
Because there are syntactically different ways to write this operation, the subgraph that implements the GELU in different models may appear in different ways.
In fact, within the different transformer models within the Huggingface (HF) Transformers Repository \cite{hf-transformers},
the division of \texttt{x} by $2$ appears both as \texttt{Div(x,2)} in some places and \texttt{Mul(x,0.5)} as others.
In cases like this where an AI compiler phrases the same conceptual operation in multiple ways, programmers can avoid an exponential blowup in the number of patterns they must write to optimize that operation by using \emph{pattern alternates}.
Figure~\ref{fig:gelu-example} shows a pattern \texttt{Half(x)} with two alternates, one which defines it as division by two, and another as mulitiplication by a half.
This pattern matches against a subgraph if either of its two alternates match.
Syntactically in \pypm{}, pattern alternates are written by writing two different patterns with the same name.
Figure~\ref{fig:gelu-example} also shows a single GELU pattern which uses \texttt{Half} for the $\frac{x}{2}$ part. This pattern matches a GELU that uses either of the variants found in the Huggingface Transformers.

Pattern alternates are tried in order that they are defined in the file. If
matching one alternate fails for any reason --- either because an
\texttt{assert} fails or a constructor does not match --- matching backtracks,
erasing any bound variables and proceeding with the next match.

\begin{figure}
\begin{verbatim}
@pattern
def Half(x):
    return Div(x,2)

@pattern
def Half(x):
    return Mul(x,0.5)

@pattern
def Gelu(x):
    return Mul(Half(x),Plus(1,Erf(Div(x,1.41...))))
\end{verbatim}
\caption{Alternate GELU Pattern}
\label{fig:gelu-example}
\end{figure}

\subsection{Recursive Patterns and Function Patterns}

One use case for \pypm{} is to fuse chains of operations into single nodes. For example, since \texttt{RELU}
is idempotent, an arbitrary chain of \texttt{RELU}s can be compressed to a single \texttt{RELU}. However,
\pypm{} as described so far cannot implement this use case, since all of the patterns described so far match subgraphs of fixed size. To match arbitrary-size graphs, we use recursive patterns.
As the name suggests, recursive patterns allow one to define a pattern that refers to itself in its own body, which lets
patterns pick out arbitrary-sized tree-shaped subgraphs from a computation graph.

The example in Figure~\ref{fig:unaryop-example} shows a pattern \texttt{UnaryChain} that matches subgraphs
that look like \texttt{RELU(RELU(RELU(...)))}. In fact, this example actually matches more subgraphs than just those:
it uses \emph{function patterns}, abstracting over the particular operation in question, and matching a
tower of a single arbitrary unary operator \texttt{f}.
Used alone, matching recursive patterns never terminates. Like normal recursion, recursive patterns need a base case.
In \pypm{}, this is accomplished with pattern alternates. In this example, the second alternate for \texttt{UnaryChain}
does not make recursive use of itself, and so once the first alternate fails to match, the matcher will backtrack and select
the second, terminating the match.

\begin{figure}
\begin{verbatim}
@pattern
def UnaryChain(x,f):
    return f(UnaryChain(x,f))

@pattern
def UnaryChain(x,f):
    return f(x)
\end{verbatim}
\caption{Recursive Unary Function Pattern}
\label{fig:unaryop-example}
\end{figure}

\subsection{Local Variables and Match Constraints}
Figure~\ref{fig:unary-binary-example} shows a pattern similar to
the unary operation chain example of Figure~\ref{fig:unaryop-example},
matching subgraphs that are composed of combinations of one binary and one unary operation.
However, there is one main difference. The variable \texttt{x} returned by a successful match
gets bound to the \emph{root} of the tree, as opposed to one of the leaves.
This is accomplished by a combination of two new \pypm{} features:
match constraints and local variables.

Local variables are straightforward. The special function \texttt{var()} generates
a fresh pattern variable, locally scoped to the pattern being defined.
For the (overarching) pattern to match, every fresh variable introduced must 
eventually be bound to \emph{some} subterm.

Like guard patterns, match constraints are another kind of assertion constraining the variables in a pattern.
The match constraint \texttt{x <= p} ensures that if \texttt{x}
gets bound to a subgraph \texttt{t} to during matching of the surrounding pattern, \texttt{t} must itself match \texttt{p}.

\begin{figure}
\begin{verbatim}
@pattern
def P(x,f,g):
    y = var()
    x <= f(P(y,f,g))
    return x

@pattern
def P(x,f,g):
    y = var()
    z = var()
    x <= g(P(y,f,g), P(z,f,g))
    return x

@pattern
def P(x,f,g):
    return x
\end{verbatim}
\caption{Recursive Pattern with Local Variables and Match Constraints}
\label{fig:unary-binary-example}
\end{figure}

\subsection{\pypm{} Implementation}
\label{subsec:pypm-impl}

\pypm{} is implemented in two components. The first half is a frontend, written as a library in
Python, that transforms the shalowly embedded syntax of \pypm{} programs into a
portable serialized binary format.  This transformation is essentially a
symbolic execution which instruments the method defintions annotated by
\texttt{@pattern} and \texttt{@rule} decorators, and then evaluates them with
specially constructed ``symbolic inputs'' which trace all paths in the method
bodies, collecting pattern constraints and returning a final pattern object to
be serialized.  These serialized pattern binaries are then dynamically loaded
into and interpreted by the second half, a C++ backend integrated into the
\dlcb{} compiler. This pattern-matching pass of \dlcb{} repeatedly walks the
nodes in an operator graph, greedily rewriting all of the patterns it can match
until no matches remain.

\pypm{} files are regular Python files, and are executed by the Python interpreter.
Importing the \texttt{pypm} library brings the \texttt{@pattern} and \texttt{@rule} decorators into scope,
and interpreting a method annotated with one of these decorators adds it to a global registry visible to the library.
The entry point to the compiler is when a script asks to serialize its rules with \texttt{pypm.serialize()}. When this happens,
the library walks all of the registered method definitions and modifies their syntax, replacing builtin python operations
with special versions that can be intercepted.
For example assertions \texttt{assert e} are turned into calls to a function \texttt{\_pattern\_assert(e)}, assignments are turned into calls to \texttt{\_pattern\_bind\_name(x,e)},
and control flow is replaced by code that will execute every branch, keeping track of which branch it's in.
The transformed code is then JIT-compiled with the python bytecode compiler and then run, passing in special \texttt{pypm.Parameter}
objects for all of its inputs. The execution of the pattern/rule traces the
function and returns a python object that is its core calculus representation. Last, 
these symbolic patterns and rules are serialized into the portable format and returned to the user, who can then dump them to a file.

When \dlcb{} starts, it dynamically loads and parses a user-specified set of pattern binaries.
Then, when the rewriting compiler pass runs on an operator graph,
the compiler repeatedly traverses the graph, attempting to match any of the patterns.
Each time a node is visited, the compiler attempts to match the subtree rooted
at that node against each of the loaded patterns, in order of their appearance
in the original python file. When a match is found, the corresponding rule (if
any) fires, and the replacement is built and substituted into the graph in place
of the subgraph the pattern matched.
The matching subroutine itself is complex (a great deal of complex C++), and requires walking both
the operator graph and the pattern. Our theoretical contribution in Section~\ref{sec:theory} contains a stylized account of this algorithm, and a proof that it agrees with
an intuitive semantics of \pypm{}.
In short, the algorithm checks that the subgraph rooted at the current node has the same structure as the pattern, and that all of the assertions (guards) evaluate to true. If any nested patterns are encountered,
\dlcb{} recursively applies this algorithm at the current node. If the pattern has alternates, \dlcb{} tries them in order, backtracking and trying the next alternative if matching fails.

\section{Theory}
\label{sec:theory}

In this section, we present \core{}, our formalization of \pypm{}.
\core{} simplifies the details of \pypm{} down to a set of core pattern primitives, abstracting away from its embedded Python syntax and complicated C++ implementation.
Moreover, computation graphs of operators are abstracted as syntax trees in \core{}, and so the calculus takes the familiar form of pattern matching against terms.
\core{} has two semantics which describe the operation of this pattern matching langauge.
The first is a \emph{declarative} semantics; an idealized description of which terms match which patterns. This semantics is highly nondeterministic,
and makes clairvoyant guesses in the presence of pattern alternates. The second is an \emph{algorithmic} semantics which gives an account of how the actual matching algorithm runs, maintaining a stack of alternatives
and backtracking when matches fail.
This setup can be thought of through the analogy of \pypm{} with logic languages. The declarative semantics
can be thought of as a proof system for pattern matching: given a witness, verify that the formula is satisfied. Meanwhile, the algorithmic semantics defines a \emph{proof search} procedure for the logic: search for a witness to the formula.

In Section~\ref{sec:simple-patterns}, we describe the minimal subset of \core{} required to understand \core{} and its two semantics.
We then layer on the rest of the constructs in the rest of Section~\ref{sec:theory}.

\subsection{Simple Patterns and Matching}
\label{sec:simple-patterns}
Because \pypm{} allows programmers to define operators, \core{} is parameterized over the set of operators in question.
In particular, fix a set of operators $\Sigma$, and $\arity{\cdot} : \Sigma \to \mathbb{N}$ denote their arities.
We use $f,g$ to range over elements of $\Sigma$. The set of terms is then inductively generated by correctly-saturated applications of elements of $\Sigma$ to a list of terms. Note that as usual, constants are function symbols of arity $0$.
Patterns are then either (a) pattern variables $x$, (b) correctly-saturated applications of function symbols to a list of patterns $f(p_1,\dots,p_n)$, or (c) pattern alternates, written $p \| p'$.
Figure~\ref{fig:terms-and-patterns} gives the grammar for both terms and basic patterns.

\begin{figure}
\begin{align*}
    t &::= f(t_1,\dots,t_n) \qquad (\arity{f} = n)\\
    p &::= x \\
    &\mid f(p_1,\dots,p_n) \qquad (\arity{f} = n)\\
    &\mid p \| p'
\end{align*}
    \caption{Grammar of Terms and Basic Patterns}
    \label{fig:terms-and-patterns}
\end{figure}

\subsubsection{Declarative Semantics}
The declarative semantics describes which terms match which patterns, under which substitutions.
The judgment form of this semantics is $\matchNOFUN{p}{t}{\theta}$~,~which we read as ``the term $t$ matches the pattern $p$ with substitution $\theta$''.
The $\theta$ position of this relation is a finite map from variables to terms.
In the logic programming analogy, the substitution $\theta$ can be thought of as a witness to the match: it proves that $t$ matches $p$ by exhibiting a unifier.

The rules of the declarative semantics are defined in Figure~\ref{fig:kernel-decsem} and proceed by cases on the pattern $p$ being matched.
The \rulename{P-Var} rule states that a variable $x$ matches against $t$ with substitution $\theta$ if $\theta$ maps $x$ to $t$.
Notationally, we write $\theta(x) \mapsto t$ to mean ``the substitution $\theta$ maps $x$ to $t$''.
For function matching, the rule \rulename{P-Fun} says that a pattern $f(p_1,\dots,p_n)$ matches with $\theta$ a term if the term looks like $f(t_1,\dots,t_n)$, and for all $i$, $p_i$ matches $t_i$ with $\theta$.
Lastly, the rules \rulename{P-Alt-1} and \rulename{P-Alt-2} says that a pattern alternate $p \| p'$ matches $t$ if either $p$ or $p'$ matches $t$.
We note that in some formalizations of pattern matching and logic programming \cite{egg,efficient-e-matching}, all closed/ground patterns are terms, and so ``the term $t$ matches against $p$ with substitution $\theta$''
can be instead phrased as $p[\theta] = t$. Because of pattern alternates (and other features we'll see later), not all ground patterns are terms, and so we must specify what it means to ``match''
with an inductive relation.

\begin{figure}
    \begin{mathpar}
        \infer[P-Var]{\theta(x) \mapsto t}{\matchNOFUN{x}{t}{\theta}}
        
        \infer[P-Fun]{
            \forall i. \left(\matchNOFUN{p_i}{t_i}{\theta}\right)\\
        }{
            \matchNOFUN{f(p_1,\dots,p_n)}{f(t_1,\dots,t_n)}{\theta}
        }
        \\
        \infer[P-Alt-1]{
            \matchNOFUN{p}{t}{\theta}
        }{
            \matchNOFUN{p \| p'}{t}{\theta}
        }

        \infer[P-Alt-2]{
            \matchNOFUN{p'}{t}{\theta}
        }{
            \matchNOFUN{p \| p'}{t}{\theta}
        }
        
        \end{mathpar}
    \caption{Declarative Semantics of Basic Patterns}
    \label{fig:kernel-decsem}
\end{figure}

\begin{theorem}[Match Weakening]
    If $\matchNOFUN{p}{t}{\theta}$ and $\theta \subseteq \theta'$ then $\matchNOFUN{p}{t}{\theta'}$
\end{theorem}

\subsubsection{Algorithmic Semantics}
The algorithmic semantics describes how to determine if a term $t$ matches a
pattern $p$, and how to build the substitution if a match exists.  This
semantics is a state transition system which builds up a substitution over time, binding variables to subterms as it traverses
the pattern and term being matched. Because of the nondeterministic nature of pattern alternates, the state machine
may need to make guesses that it can un-make if it decides they were the wrong choice. To this end,
the algorithmic semantics is backtracking-based, keeping track of a stack of saved program points to jump to if a conflict is encountered.

Formally, the state of the machine is one of three options. There are two terminal states, $\successNOFUN{\theta}$ recording a successful match terminating with $\theta$,
or $\failure$ indicating a failure to match. The most important state is $\runningNOFUN{\theta}{stk}{k}$, which records the current substitution $\theta$,
a stack $stk$ for backtracking, and a continuation $k$ of actions to take before the matching can be considered successful.

A continuation $k$ is list of \emph{actions}, which (for the moment), are directives $\actDoMatch{p}{t}$ to try matching $p$ against $t$.
To run the algorithmic semantics on $p$ and $t$, we simply set the continuation to $\left[\actDoMatch{p}{t}\right]$. As the semantics runs,
the continuation grows and shrinks as matching obligations are generated and discharged.
The stack $stk$ is a list of backtrack nodes $\btNODENOFUN{\theta}{k}$, recording a saved substitution and a continuation at a choice point.

The state transition system that defines the algorithmic semantics is given by a
small-step relation $\stateTrans{st}{st'}$ whose rules can be found in
Figure~\ref{fig:kernel-algsem}.
The step relation inspects the head of the continuation and transitions accordingly.
If the head of the continuation is a command to match a variable $x$ against a term $t$, there are three options.
First, in \rulename{ST-Match-Var-Bind}, the current substitution $\theta$ does not yet have a binding for $x$, in which case
we step to a state where the substitution has been extended by $\{x \mapsto t\}$.
Next, if $\theta$ already maps $x$ to $t$, \rulename{ST-Match-Var-Bound} says to simply continue.
Last, in \rulename{ST-Match-Var-Conflict}, $\theta$ maps $x$ to a term other than $t$, and so we cannot proceed further along this course.
In this case, if the backtracking stack is empty, we step to the $\failure$ state. Otherwise we pop the backtracking stack
and replace the current substitution and continuation with the saved state.
This dispatch (determining if the stack is empty and stepping to the correct state) is handled by a metafunction $\backtrack{stk}$
to condense the two options into one rule.

If the head of the continuation says we must match a function pattern $f(p_1,\dots,p_n)$ against a term $t$,
we inspect the term. If $t$ is of the form $g(t_1,\dots,t_m)$ for some $f \neq g$ or if $m \neq n$, then we attempt to backtrack with \rulename{ST-Match-Fun-Conflict}.
Otherwise, $t = f(t_1,\dots,t_n)$ and so \rulename{ST-Match-Fun} applies and we add the instructions to match $p_i$ with $t_i$ to the continuation.

Last, if the head of the continuation says we must match a pattern alternate $p
\| p'$ against $t$, \rulename{ST-Match-Alt} applies. In this rule we try the
first alternate $p$, stepping to a state where the head of the continuation says
to match $p$ against $t$. We also push onto the backtracking stack,
recording the current substitution and the current continuation with the second alternative $p'$ prepended.
This ensures that if matching $p$ against $t$ fails (or generates substitutions that conflict with the current continuation),
we can try the other route.

Note that because this algorithmic treatment of pattern alternates is deterministic and left-eager,
the algorithmic semantics is necessarily not complete with respect to the declarative semantics: matching $t = f(c_1,c_2)$ against $p = f(x,y) \| f(y,x)$
with the algorithmic semantics will only ever yield the substitution $\{x \mapsto c_1,y\mapsto c_2\}$, despite the fact that the judgment
$\matchNOFUN{p}{t}{\{x \mapsto c_2,y\mapsto c_2\}}$ is also derivable.
However the algorithm is sound in the sense that if the algorithmic semantics runs on $p$ and $t$ to $\successNOFUN{\theta}$
then $\matchNOFUN{p}{t}{\theta}$, and if it runs to $\failure$, no such $\theta$ exists.

\begin{figure*}
    \begin{align*}
        a &::= \actDoMatch{p}{t}\\
        k &::= \contNil \mid \contCons{a}{k}\\
        stk &::= \stkEmp \mid \btNOFUN{stk}{\theta}{k}\\
        st &::= \successNOFUN{\theta} \mid \failure \mid \runningNOFUN{\theta}{stk}{k}\\
        \\
    \backtrack{\stkEmp} &= \failure\\
    \backtrack{\btNOFUN{stk}{\theta}{k}} &= \runningNOFUN{\theta}{stk}{k}
    \end{align*}

    \begin{mathpar}
        
    \infer[ST-Success]{ }{
        \stateTrans{
            \runningNOFUN{\theta}{stk}{\contNil}
        }{
            \successNOFUN{\theta}
        }
    }

    \infer[ST-Match-Var-Bind]{
        \neg \exists t'. \left(\theta(x) \mapsto t'\right)
    }{
        \stateTrans{
            \runningNOFUN{\theta}{stk}{\contCons{\actDoMatch{x}{t}}{k}}
        }{
            \runningNOFUN{\theta \cup \{x \mapsto t\}}{stk}{k}
        }
    }

    \infer[ST-Match-Var-Bound]{
        \theta(x) \mapsto t 
    }{
        \stateTrans{
            \runningNOFUN{\theta}{stk}{\contCons{\actDoMatch{x}{t}}{k}}
        }{
            \runningNOFUN{\theta}{stk}{k}
        }
    }

    \infer[ST-Match-Var-Conflict]{
        \theta(x) \mapsto t'\\
        t' \neq t
    }{
        \stateTrans{
            \runningNOFUN{\theta}{stk}{\contCons{\actDoMatch{x}{t}}{k}}
        }{
            \backtrack{stk}
        }
    }

    \infer[ST-Match-Fun]{
        k' = \left[\actDoMatch{p_1}{t_1},\dots,\actDoMatch{p_n}{t_n}\right]\\
    }{
        \stateTrans{
            \runningNOFUN{\theta}{stk}{\contCons{\actDoMatch{f(p_1,\dots,p_n)}{f(t_1,\dots,t_n)}}{k}}
        }{
            \runningNOFUN{\theta}{stk}{k' \texttt{++} k}
        }
    }

    \infer[ST-Match-Fun-Conflict]{
        f \neq g \vee m \neq n\\
    }{
        \stateTrans{
            \runningNOFUN{\theta}{stk}{\contCons{\actDoMatch{f(p_1,\dots,p_n)}{g(t_1,\dots,t_m)}}{k}}
        }{
            \backtrack{stk}
        }
    }

    \infer[ST-Match-Alt]{
        stk' = \btNOFUN{stk}{\theta}{\contCons{\actDoMatch{p'}{t}}{k}}
    }{
        \stateTrans{
            \runningNOFUN{\theta}{stk}{\contCons{\actDoMatch{p \| p'}{t}}{k}}
        }{
            \runningNOFUN{\theta}{stk'}{\contCons{\actDoMatch{p}{t}}{k}}
        }
    }

    \end{mathpar}

    \caption{Algorithmic Semantics of Basic Patterns}
    \label{fig:kernel-algsem}
\end{figure*}

\begin{theorem}[Algorithmic Soundness]
    If $$\runningNOFUN{\emptyset}{\stkEmp}{\left[\actDoMatch{p}{t}\right]} \mapsto^+ \successNOFUN{\theta}$$
    then $\matchNOFUN{p}{t}{\theta}$. If
    $$\runningNOFUN{\emptyset}{\stkEmp}{\left[\actDoMatch{p}{t}\right]} \mapsto^+ \failure$$
    then there are no $\theta$ such that $\matchNOFUN{p}{t}{\theta}$
\end{theorem}
\begin{proof}
    A mechanized proof of soundness of \core{} can be found in the supplementary materials in \texttt{Proof.v},
    theorems \texttt{succ\_sound} and \texttt{fail\_sound}.
\end{proof}

\subsection{Guarded Patterns}

\begin{figure}
    \begin{align*}
        e &::= t.\alpha \mid x.\alpha \mid e + e \mid e - e \mid \dots \\
        g &::= e = e' \mid e < e' \mid g \wedge g' \mid g \vee g' \mid \neg g\\
        p &::= \dots \mid \guardPat{p}{g}
    \end{align*}

    \begin{mathpar}
     \inferrule[P-Guard]{
            \matchNOFUN{p}{t}{\theta}\\
            \scott{g[\theta]} = \texttt{True}
        }{
            \matchNOFUN{\guardPat{p}{g}}{t}{\theta}
        }   
    \end{mathpar}

    \caption{Guarded Pattern Syntax and Declarative Semantics}
    \label{fig:guard-pat-syn}
\end{figure}

The analogue of \pypm{}'s \texttt{assert} feature in \core{} is \emph{guarded patterns}
$\guardPat{p}{g}$.  In short, $\guardPat{p}{g}$ matches against $t$ under $\theta$ if $p$ matches against $t$ under $\theta$, and the
assignment of variables in $\theta$ makes the guard term $g$ true, modelling assertion checks.

Guards, written $g$, are boolean constraints over arithmetic expressions. Those arithmetic expressions
can include terms like $x.\alpha$, where $x$ is a pattern variable, and $\alpha$ is an
\emph{attribute}. While \pypm{} defines a concrete set of tensor-specific attributes like \texttt{.shape} \texttt{.rank}
and \texttt{.eltType},
we leave \core{} abstract and instead require a fixed set of attributes $\mathcal{A}$, along with an 
interpretation function $\scott{\cdot} : \mathcal{A} \to \texttt{Term} \rightharpoonup \mathbb{N}$ defining their meaning on terms.
We then lift this function to a $\mathbb{N}$-valued denotation on closed arithmetic expressions by $\scott{t.\alpha} = \scott{\alpha}(t)$ (and compositionally otherwise),
and subsequently lift it to boolean-valued denotation on closed guard terms.
Figure~\ref{fig:guard-pat-syn} shows how we extends the syntax of patterns to add guard terms,
defines the syntax of guard terms and expressions, and defines the declarative semantics of guarded patterns.
The rule \rulename{P-Guard} simply says that $\guardPat{p}{g}$ matches against $t$ if $p$ matches against $t$,
and if the substitution instance $g[\theta]$ is deemed \texttt{True} by the guard semantics.

The algorithmic semantics of guarded patterns can be found in \theappendix{}.

\subsection{Existiential Variables and Match Constraints}
\pypm{}'s \texttt{var()} and \texttt{<=} (match constraint) features also have corresponding constructs in \core{}.
The former is modeled by \emph{existential patterns} $\exists x.p$, which bring a pattern variable $x$ into scope.
The latter is modeled by a second kind of guarded pattern for match constraints, which we write as $\backMatch{p}{x}{p'}$.
The semantics of both constructs are straightforward. $\exists x.p$ matches against $t$ under $\theta$ if there exists some $t'$
such that $p$ matches against $t$ under $\theta \cup \{x \mapsto t'\}$.
Meanwhile, $\backMatch{p}{x}{p'}$ matches against $t$ under $\theta$ if $p$ matches against $t$, and $\theta(x)$ matches against $p'$.
These two constructs highlight the highly nondeterministic nature of the \core{} semantics, and demonstrate that our algorithmic semantics is nontrivial. Both \rulename{P-Exists}
and \rulename{P-MatchConstr} invent a term $t'$ from nowhwere, and so are obviously not directly implementable.

\begin{figure}
    \begin{align*}
        p &::= \dots \mid \exists x.p \mid \backMatch{p}{x}{p'}
    \end{align*}

    \begin{mathpar}
     \inferrule[P-Exists]{
        \matchNOFUN{p}{t}{\left(\theta \cup \{x \mapsto t'\}\right)}
    }{
        \matchNOFUN{\exists x.p}{t}{\theta}
    }   

    \inferrule[P-MatchConstr]{
        \matchNOFUN{p}{t}{\theta}\\
        \theta(x) \mapsto t'\\
        \matchNOFUN{p'}{t'}{\theta}
    }{
        \matchNOFUN{\backMatch{p}{x}{p'}}{t}{\theta}
    }
    \end{mathpar}

    \caption{Syntax and Declarative Semantics of Existential Variables and Match Constraints}
    \label{fig:ex-backmatch-syn}
\end{figure}

The algorithmic semantics for existential patterns and match constraints can be found in \theappendix{}.

\subsection{Function Variables}
In Section~\ref{sec:examples}, we saw a \pypm{} pattern that matched a use of an arbitrary binary operation. So far, \core{} has only
allowed for patterns like $f(p_1,\dots,p_n)$, where $f \in \Sigma$ is some particular function symbol. In this section, we handle this use case
by introducing a second kind of variable called \emph{function variables} that range over function symbols rather than whole terms.
Then, by using a pattern variable $F$, we can write the pattern $F(x,y)$ to match a use of any binary operation.
Formally, we extend the grammar of patterns so that if $F$ is a function variable, we can form the pattern $F(p_1,\dots,p_n)$.
Function variables have the same semantics as regular pattern variables, but simply bind function symbols rather than terms. For example, the pattern $F(F(x))$ matches terms that look like a unary function applied to itself twice.
As discussed when function patterns were introduced, this feature takes \core{} beyond the standard logical strength of either (a) pattern matching in functional languages, or (b) logic programming.

As such, this feature requires a slight modification of what matching means.
Witnessing that $t$ matches $p$ does not just require a substitution describing which pattern variables map to which terms, it also requires a map explaining which function variables get mapped to which symbols.
To this end, we extend the matching judgment of the declarative semantics to include a function substitution $\phi$, a finite map from function variables to elements of the set $\Sigma$.
We write this new judgment as $\match{p}{t}{\theta}{\phi}$. All of the previously-described rules of the declarative semantics do not touch this component; it simply goes along for the ride. To see the full declarative rules of \core{}
rendered with the function substitution included, see \theappendix{}.
The declarative semantics rule for function variable application patterns \rulename{P-Fun-Var} is straightforward: for $f(t_1,\dots,t_n)$ to match $F(p_1,\dots,p_n)$,
each of the $t_i$ must match $p_i$, and $\phi$ must map $F$ to $f$.

The algorithmic semantics undergoes a similar modification where we add a function substitution $\phi$ to the running state, success state, and to backtracking nodes on the stack.
The algorithmic rules for matching function variable patterns behave like a combination of the \rulename{ST-Match-Fun} rule for match actions of standard function patterns and the \rulename{ST-Match-Var-*} rules for
matching a variable. The full algorithmic semantics for function variables can be found in \theappendix{}.

\subsection{Recursive Patterns}

\pypm{}'s pattern recursion is modeled in \core{} by fixpoint definitions,
written $\mu P(x).p$. A pattern of this form may refer to the ``recursive pattern call''
$P(x)$ anywhere within the body pattern $p$. To match a term $t$ against a recursive pattern, we simply unfold the recursion
one step, substituting $\mu P(x).p$ for all ocurrences of $P(x)$ inside $p$ to get a new term $p'$,
and then match $t$ against $p'$.

Note that this includes the possibility of nonterminating patterns, as $\mu P(x). P(x)$
immediately unfolds to itself. A full formal treatment of the declarative and algorithmic semantics
of recursive patterns can be found in \theappendix{}.








\section{Use Of \pypm{} and Evaluation}
\label{sec:eval}

In this section we review and evaluate two uses of \pypm{} --- its primary use case in optimization and
instruction selection, and a more novel use case we call directed graph partitioning.

\subsection{Optimization}

To demonstrate \pypm{}'s primary use case --- rule-based optimization and instruction selection in AI models --- we'll show how rewriting by hand-crafted
patterns for common operatons can produce significant performance improvements across the board in a wide range of models.

All of the most advanced large language models are based on the transformer architecture, which
has \emph{multi-head attention} (MHA) as its most basic building block \cite{attention}. MHA is the operation
$\texttt{MHA}(Q,K,V) = \text{softmax}\left(\alpha QK^T\right)V$, where the softmax is applied row-wise and $\alpha$ is a constant.
Often, this operation is implemented
verbatim, as three matrix products, a transpose, and a row-wise softmax. For an operation that is so frequently used during transformer model inference,
this is clearly sub-optimal. Instead, we have developed a custom fused kernel that efficiently implemenents the forward pass of MHA on a GPU.
The supplementary materials include a simplified version of a pattern that matches instances of MHA in tensor computations, and a rule that replaces
them with a call to fused multi-head attention (FMHA).

Because of the structure of neural networks, another common operation in inference is applying a pointwise activation function like RELU
or GELU \cite{gelu} to the result of a matrix multiplication or convolution operation. To this end, we have developed a GEMM matrix multiplication operation
can include an \emph{epilog}: a pointwise activation function fused into the implemenation of the matrix multiplication~\footnote{This kernel is similar to, but not the same as, the cuBLAS GEMM kernel discussed in Section~\ref{sec:intro}.}. Of course,
as discussed in Section~\ref{sec:intro}, targeting this sort of operation is complicated and requires checking that the particular GEMM instance has size and shape paramters supported by the kernel.
The supplementary materials also show a (once again simplified) pattern that matches fusable epilogs, and a rule that rewrites them to uses of the library operation.

We benchmark these optimizations with two different benchmarking suites. The first is Huggingface's transformers benchmark \cite{hf-transformers},
which tests the performance of inference in a wide range of pre-trained transformer models. The second is the TorchVision (TV) benchmark, which
tests the performance of inference in a large set of pre-trained computer vision models. Both benchmarks include scripts to run a large number of 
inference passes through all of their models, and then report statistics about performance. This setup is intended to benchmark the same implementations
of a wide range of models on different hardware. Here, we use repurpose these benchmarking scripts: we keep the hardware \emph{fixed} \footnote{All benchmarks are run on a machine with a single Nvidia RTX A6000 GPU.}
and instead vary the models themselves by optimizing them using \dlcb{} with different \pypm{} rewrites enabled.

In both cases the model benchmarks are written as code in PyTorch \cite{pytorch}, a popular Python framework for writing AI models. Our benchmarking harness first ingests the models
and uses PyTorch to compile them to PyTorch's internal IR. We then take the PyTorch IR code, transform it to \dlcb{}'s internal IR, run the \dlcb{} compiler, and dump the result back to PyTorch IR.
Because \dlcb{} only understands a (large) subset of PyTorch operators, unfamiliar operators are represented as opaque nodes, and cannot be matched by the \pypm{} pass.
Last, we pass the optimized PyTorch code back to the HuggingFace and TorchVision benchmark scripts. These benchmark scripts (written by the HuggingFace and PyTorch teams, respectively) repeatedly
run the models in forward (inference) mode, and report an average per-iteration wall-clock time.

We use \dlcb{} to compile each model in the two benchmarks four ways. Once with the FMHA and Epilog optimizations disabled,
once each with FMHA and Epilog only, and once with both optimizations enabled simultaneously.
The results of the benchmarks are shown in Figures~\ref{fig:hf-benchmark}~and~\ref{fig:tv-benchmark}.
These figures show histograms, reporting the distributions of \emph{relative speedups} (when compared to \dlcb{} with neither optimization enabled.) across all models achieved
under each set of optimizations. 

All models were compiled with \dlcb{} and subsequently benchmarked on a machine with an AMD Ryzen 7 3700X 8-Core processor clocked at 3630MHz, 128GB RAM, and an Nvidia A6000 Ampere GPU.
The \dlcb{} compiler itself runs entirely on the CPU. 
Figures~\ref{fig:hf-compile-time}~and~\ref{fig:tv-compile-time} contain a graphs of the time spent running the pattern matcher
during \dlcb{} evaluation as a function of number of matches that are found in a model.
While the number of matches is the primary determiner of compile-time cost, time spent matching also depends on the size of the AST of each model.
Even in cases where matches are found, the pattern matcher must still traverse the entire model, potentially spending time checking ``partial matches'' that don't end up
actually matching the pattern. This is borne out by the Epilog matching costs: Even when there are none, 
the implementation takes 2 orders of magnitude longer looking for Epilog matches than MHA matches 
because there are \emph{many} more matrix multiplies in all of the HF and TV models than potential MHA matches.
However, we find that this compile-time cost is minimal across the board, with none of the matching passes --- between both the FMHA and Epilog patterns on the HF and TV models --- ever taking longer than 3 seconds
to run to fixpoint on a particular model.




\begin{figure}
    \includegraphics[width=0.4\textwidth]{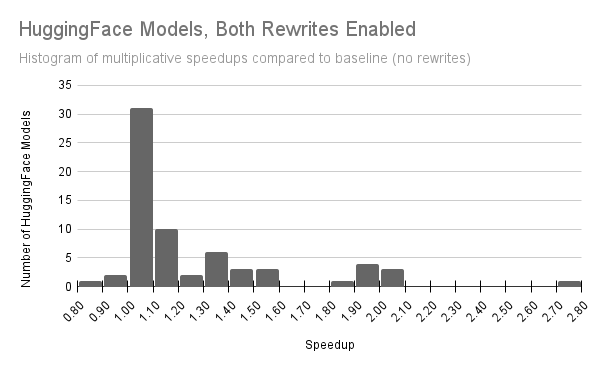}
    \includegraphics[width=0.4\textwidth]{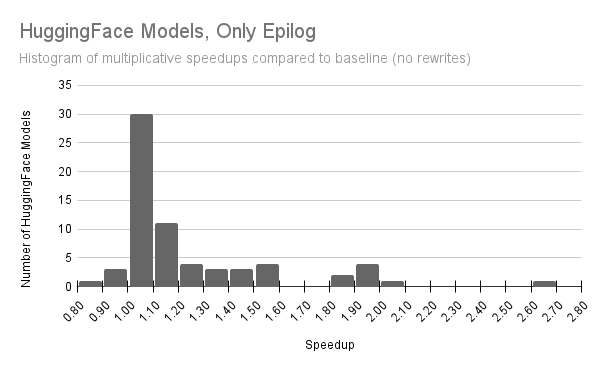}
    \includegraphics[width=0.4\textwidth]{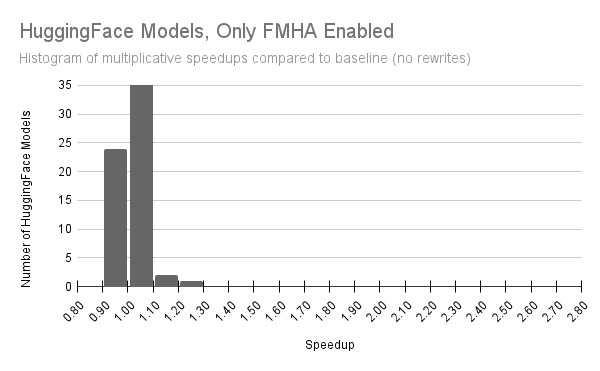}
    \caption{HuggingFace Benchmarks}
    \label{fig:hf-benchmark}
\end{figure}

\begin{figure}
    \includegraphics[width=0.4\textwidth]{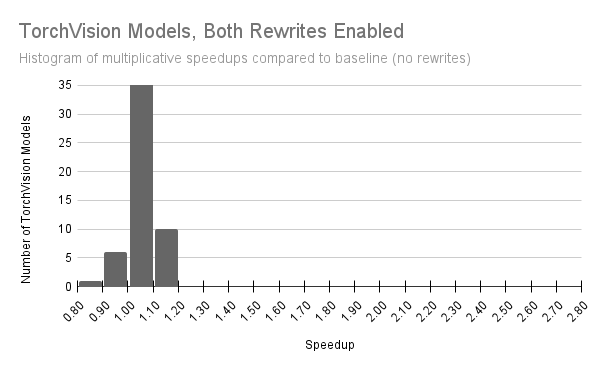}
    \includegraphics[width=0.4\textwidth]{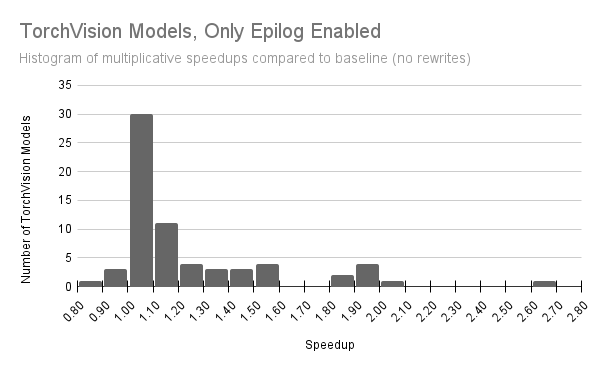}
    \includegraphics[width=0.4\textwidth]{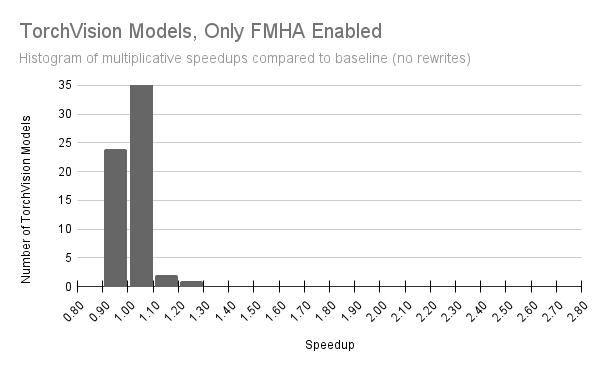}
    \caption{TorchVision Benchmarks}
    \label{fig:tv-benchmark}
\end{figure}

\begin{figure}
    \includegraphics[width=0.4\textwidth]{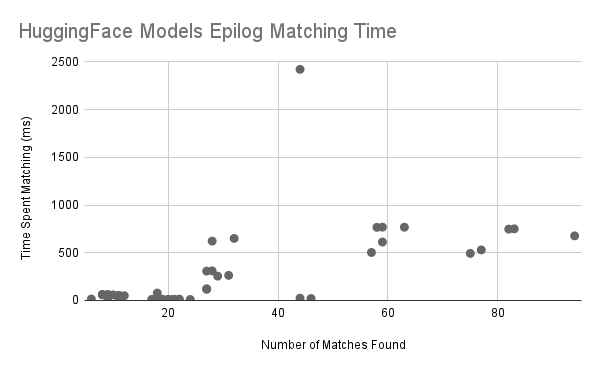}
    \includegraphics[width=0.4\textwidth]{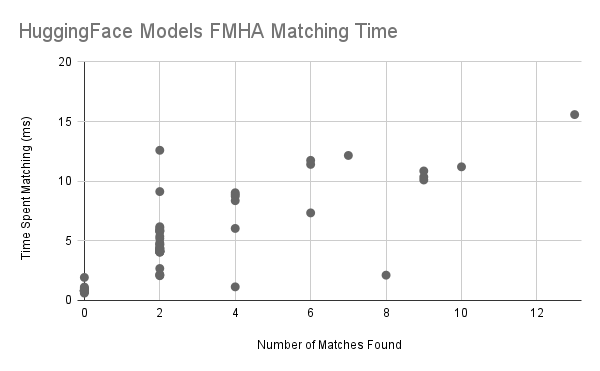}
    \caption{HuggingFace Compile Time Cost}
    \label{fig:hf-compile-time}
\end{figure}

\begin{figure}
    \includegraphics[width=0.4\textwidth]{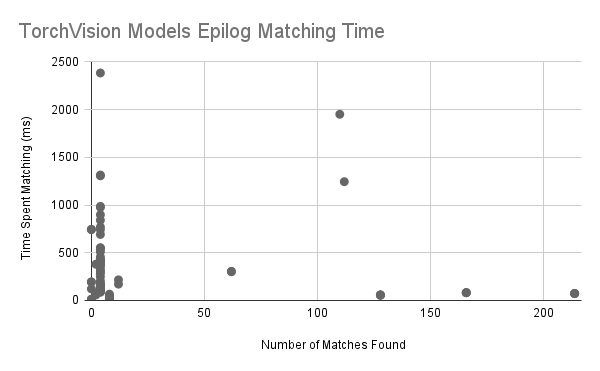}
    \includegraphics[width=0.4\textwidth]{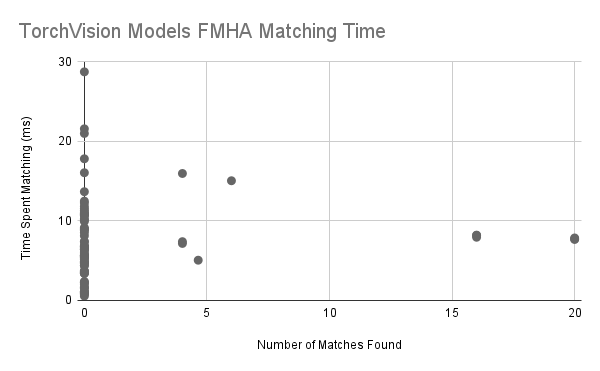}
    \caption{TorchVision Compile Time Cost}
    \label{fig:tv-compile-time}
\end{figure}

We note that it is not a contribution of this work --- nor is it particularly surprising ---
to demonstrate that replacing inefficient implementations with bespoke fused ones yields a performance improvement.
The point of these benchmarks is instead to demonstrate that \pypm{} is sufficiently expressive to 
be used by experts in handcrafting effective optimizations.

\subsection{Directed Graph Partitioning}
So far we have described \pypm{} alongside its primary use case, namely to replace subgraphs of tensor computations with hand-crafted kernels in cases where optimizing AI compilers cannot yet recognize the opportunities
to do so. While we have shown \pypm{} to be effective in this use case, it does have one primary drawback which we freely recognize: hand-crafted replacements msut be available for this technique to apply! While \pypm{} features like pattern alternates
and guarded rules simplify the challenge of crafting replacements, it may not be feasible to manually write rules for very complex patterns which \emph{do} have optimized replacments.
Figure~\ref{fig:pw-subgraph} shows an example of such a pattern: \texttt{MatMulEpilog} describes subgraphs that look like a matrix multiply followed by some number of pointwise operations. An optimizing AI compiler should be able to
fuse such a subgraph into a single kernel that both performs the matrix multiplication and also applies the all of the pointwise operations. However, it's not feasible to handwrite a \pypm{} rule as the right hand side --- there are simply 
too many options. A natural idea in this scenario is to create the right hand side rule ``just in
time''. When we find a match, we can pass the subgraph off to an AI compiler
that \emph{can} build the fused kernel, and use the resulting blob as the replacement.
We call this concept \emph{Directed Graph Partitioning}. By using \pypm{}
patterns, \dlcb{} can partition a computation graph into subgraphs that we know
can be optimized, and then recursively compile them.





\begin{figure}
\begin{verbatim}
@pattern
def PwSubgraph(x):
    UnaryOp = Op(1, 1)
    assert (
        UnaryOp.op_class == opclass("unary_pointwise")
    )
    y = var()
    x <= UnaryOp(PwSubgraph(y))
    return x

@pattern
def PwSubgraph(x):
    return x

def MatMulEpilog(x):
    a = var()
    b = var()
    x <= PwSubgraph(MatMul(a, b))
    return x
\end{verbatim}
\caption{Matrix Multiplication Epilog Patterns}
\label{fig:pw-subgraph}
\end{figure}




\section{Related Work}
\label{sec:related-work}
\pypm{} matches patterns in tensor computation graphs in a generic way, mediated by
a general computation graph data structure in \dlcb{}.
Of course, these computation graphs come from a particular tensor programming
framework.  \dlcb{} supports PyTorch \cite{pytorch}, JAX \cite{jax} and TensorFlow
\cite{tensorflow} backends, and so \pypm{}'s design reflects the our desire to optimize
models from those frontends in particular.

Many projects in the AI compilers space have also attempted to tackle
the more general ``retargeting'' problem that is \pypm{}'s raison d'etre.
One notable example of research in this line is Exo~\cite{exo}, which 
externalizes (as the name suggests) the process of platform-specific code generation
to libraries, rather than a compiler pass.
Another project in the area is TVM~\cite{tvm}, which is a full-fledged AI compiler
with multiple hardware backends. Unlike Exo, TVM's code generation is more traditional, and contained within the compiler
itself. XLA \cite{xla} is another project similar in scope to \dlcb{} which aims to optimize the model code
generated by the afformented deep-learning frontend frameworks.
All of the above projects have goals different than those of this paper. Rather than
arguing for our particular approach to solving the architecture targeting problem,
our focus is instead on presenting and arguing for the \pypm{} langauage itself,
and presenting its formalized core.

Some machine learning model programming frameworks like Halide~\cite{halide} and
TVM\cite{tvm} include scheduling languages that let the programmer specify a
tensor computation separately from the way that loops are ordered and memory is
laid out. The point of a scheduling langauge is to let the programmer manually
specify the way that a computation is executed on a device. Meanwhile, \dlcb{}
exists to take naively-generated model code (i.e. without a schedule), and use
rewriting combined with other techniques to schedule the computation. These are opposing solutions to similar problems --- one manual, the other automatic.
However, one could in principle use \dlcb{} in concert with a scheduling language to further optimize the scheduled code.

The idea of using rewrite rules to produce high-performance array programs or
GPU code is similarly well-explored. Projects like those of \citet{mm-beyond} and \citet{hp-functional}
use rewriting to generate GPU code and optimize existing array codes, respectively.
More recently, the success of the Egg project~\cite{egg,egglog}
has spawned a flury of works that use E-matching and equality saturation to optimize or generate code for tensor programs
\cite{tensat,diospyros,latent-idom}.

Prior work has also provided mathematical formalizations of other pattern-matching languages.
The most closely related such effort was undertaken by \citet*{efficient-e-matching}, who describe an
algorithmic virtual machine for E-matching --- a subset of \pypm{}'s matching algorithm --- and prove it
equivalent to a declarative definition of E-matching. The proof is not formalized, as the focus
of the work is on the operation of the virtual machine. As far as mechanized proof efforts of pattern
matching go, we are aware only of the work of \citet{noe-js}, who formalize the operation of JavaScript regular expression matching in Coq.

The design of \pypm{} takes inspiration from two main places. The first
is structural pattern matching as popularized by functional languages in the ML
and LISP families, where algebraic data values can be deconstructed by a
matching operator that dispatches on their shape. This pattern matching differs from \pypm{} in a
two important ways, namely that they (a) do not \emph{search} for matches as \pypm{}'s matcher does,
and (b) are required to execute efficiently (at program runtime), while \pypm{}'s pattern execution can be
somewhat inefficient due to it running at compiletime. This freedom to run (relatively) inefficiently
means that \pypm{} patterns are much more expressive as they can include alternation and recursion.
The second inspiration is from logic programming languages like Prolog \cite{prolog}, Datalog \cite{datalog}, and Egglog \cite{egglog}. As discussed in Section~\ref{sec:intro},
the concept of recursive patterns is borrowed from recursive formulae, and pattern alternates are rules with multiple disjunctive clauses.
Some \pypm{} features like function patterns fall outside standard logic programming:
function variables correspond to second-order logic programming \cite{ho-logic-programming}.
Again, \pypm{} is distinguised from a logic language because its rewriting is destructive, and so evaluation is non-monotone.

\begin{acks}
The authors would like to thank the following people for their valuable input and other contributions to
this work: Aravind Acharya, Somashekar Bhaskaracharya, William
Brandon, Cade Brown, Hanfeng Chen, Evghenii Gaburov, Bastian Hagedorn, Sean Lee
and Ben Schreiber.
\end{acks}

\bibliographystyle{ACM-Reference-Format}
\bibliography{refs}

\ifextended
\newpage
\appendix

\section{Technical Appendix}
\label{app:technical}

The full grammar of terms and patterns is found in
Figure~\ref{fig:all-terms-and-patterns}. A listing of the full declarative
semantics of \pypm{} is found in Figure~\ref{fig:all-decsem}, and a listing of
the full algorithmic semantics is in Figure~\ref{fig:all-algsem-1} and Figure~\ref{fig:all-algsem-2}.

\begin{figure*}
    \begin{align*}
        t &::= f(t_1,\dots,t_n) \qquad (\arity{f} = n)\\
        e &::= t.\alpha \mid x.\alpha \mid e + e \mid e - e \mid \dots \\
        g &::= e = e' \mid e < e' \mid g \wedge g' \mid g \vee g' \mid \neg g\\
        p &::= x \\
        &\mid f(p_1,\dots,p_n) \qquad (\arity{f} = n)\\
        &\mid p \| p'\\
        &\mid \guardPat{p}{g}\\
        &\mid \exists x.p\\
        &\mid \backMatch{p}{x}{p'}\\
        &\mid F(p_1,\dots,p_n)\\
        &\mid \mu P(x_1,\dots,x_n)[y_1,\dots,y_n]\\
        &\mid P(x_1,\dots,x_n)
    \end{align*}
        \caption{Grammar of Terms and Patterns}
        \label{fig:all-terms-and-patterns}
\end{figure*}

\begin{figure*}
    \begin{mathpar}
        \infer[P-Var]{\theta(x) \mapsto t}{\match{x}{t}{\theta}{\phi}}
        
        \infer[P-Fun]{
            \forall i. \left(\match{p_i}{t_i}{\theta}{\phi}\right)\\
        }{
            \match{f(p_1,\dots,p_n)}{f(t_1,\dots,t_n)}{\theta}{\phi}
        }
        
        \infer[P-Alt-1]{
            \match{p}{t}{\theta}{\phi}
        }{
            \match{p \| p'}{t}{\theta}{\phi}
        }

        \infer[P-Alt-2]{
            \match{p'}{t}{\theta}{\phi}
        }{
            \match{p \| p'}{t}{\theta}{\phi}
        }

        \inferrule[P-Guard]{
            \match{p}{t}{\theta}{\phi}\\
            \scott{g[\theta]} = \texttt{True}
        }{
            \match{\guardPat{p}{g}}{t}{\theta}{\phi}
        }   

        \inferrule[P-Exists]{
            \match{p}{t}{\left(\theta \cup \{x \mapsto t'\}\right)}{\phi}
        }{
            \match{\exists x.p}{t}{\theta}{\phi}
        }   
    
        \inferrule[P-MatchConstr]{
            \match{p}{t}{\theta}{\phi}\\
            \theta(x) \mapsto t'\\
            \match{p'}{t'}{\theta}{\phi}
        }{
            \match{\backMatch{p}{x}{p'}}{t}{\theta}{\phi}
        }

        \inferrule[P-Fun-Var]{
            \phi(F) \mapsto f\\
            \forall i \left(\match{p_i}{t_i}{\theta}{\phi}\right)
        }{
            \match{F(p_1,\dots,p_n)}{f(t_1,\dots,t_n)}{\theta}{\phi}
        }

        \inferrule[P-Mu]{
            \match{p[\mu P(x_1,\dots,x_n).p/P][y_i/x_i]}{t}{\theta}{\phi}
        }{
            \match{\mu P(x_1,\dots,x_n)[y_1,\dots,y_n].p}{t}{\theta}{\phi}
        }
        \end{mathpar}
    \caption{Declarative Semantics}
    \label{fig:all-decsem}
\end{figure*}

\begin{figure*}
    \begin{align*}
        a &::= \actDoMatch{p}{t} \mid \actCheckGuard{g} \mid \actUnbind{x} \mid \actBackMatch{x}{p}
        \\
        k &::= \contNil \mid \contCons{a}{k}\\
        stk &::= \stkEmp \mid \bt{stk}{\theta}{\phi}{k}\\
        st &::= \success{\theta}{\phi} \mid \failure \mid \running{\theta}{\phi}{stk}{k}\\
        \\
    \backtrack{\stkEmp} &= \failure\\
    \backtrack{\bt{stk}{\theta}{\phi}{k}} &= \running{\theta}{\phi}{stk}{k}
    \end{align*}

    \begin{mathpar}
    \infer[ST-Success]{ }{
        \stateTrans{
            \running{\theta}{\phi}{stk}{\contNil}
        }{
            \success{\theta}{\phi}
        }
    }

    \infer[ST-Match-Var-Bind]{
        \neg \exists t'. \left(\theta(x) \mapsto t'\right)
    }{
        \stateTrans{
            \running{\theta}{\phi}{stk}{\contCons{\actDoMatch{x}{t}}{k}}
        }{
            \running{\theta \cup \{x \mapsto t\}}{\phi}{stk}{k}
        }
    }

    \infer[ST-Match-Var-Bound]{
        \theta(x) \mapsto t 
    }{
        \stateTrans{
            \running{\theta}{\phi}{stk}{\contCons{\actDoMatch{x}{t}}{k}}
        }{
            \running{\theta}{\phi}{stk}{k}
        }
    }

    \infer[ST-Match-Var-Conflict]{
        \theta(x) \mapsto t'\\
        t' \neq t
    }{
        \stateTrans{
            \running{\theta}{\phi}{stk}{\contCons{\actDoMatch{x}{t}}{k}}
        }{
            \backtrack{stk}
        }
    }

    \infer[ST-Match-Fun]{
        k' = \left[\actDoMatch{p_1}{t_1},\dots,\actDoMatch{p_n}{t_n}\right]\\
    }{
        \stateTrans{
            \running{\theta}{\phi}{stk}{\contCons{\actDoMatch{f(p_1,\dots,p_n)}{f(t_1,\dots,t_n)}}{k}}
        }{
            \running{\theta}{\phi}{stk}{k' \texttt{++} k}
        }
    }

    \infer[ST-Match-Fun-Conflict]{
        f \neq g \vee m \neq n\\
    }{
        \stateTrans{
            \running{\theta}{\phi}{stk}{\contCons{\actDoMatch{f(p_1,\dots,p_n)}{g(t_1,\dots,t_m)}}{k}}
        }{
            \backtrack{stk}
        }
    }

    \infer[ST-Match-Alt]{
        stk' = \bt{stk}{\theta}{\phi}{\contCons{\actDoMatch{p'}{t}}{k}}
    }{
        \stateTrans{
            \running{\theta}{\phi}{stk}{\contCons{\actDoMatch{p \| p'}{t}}{k}}
        }{
            \running{\theta}{\phi}{stk'}{\contCons{\actDoMatch{p}{t}}{k}}
        }
    }

    \infer[ST-Match-Guard]{ }{
        \stateTrans{
            \running{\theta}{\phi}{stk}{\contCons{\actDoMatch{\guardPat{p}{g}}{t}}{k}}
        }{
            \running{\theta}{\phi}{stk}{\contCons{\actDoMatch{p}{t}}{\contCons{\actCheckGuard{g}}{k}}}
        }
    }

    \infer[ST-CheckGuard-Continue]{
        \scott{g[\theta]} = \texttt{True}
    }{
        \stateTrans{
            \running{\theta}{\phi}{stk}{\contCons{\actCheckGuard{g}}{k}}
        }{
            \running{\theta}{\phi}{stk}{k}
        }
    }
   
    \infer[ST-CheckGuard-Backtrack]{
        \scott{g[\theta]} = \texttt{False}
    }{
        \stateTrans{
            \running{\theta}{\phi}{stk}{\contCons{\actCheckGuard{g}}{k}}
        }{
            \backtrack{stk}
        }
    }

    \infer[ST-CheckName]{
        \theta(x) \mapsto t
    }{
        \stateTrans{
            \running{\theta}{\phi}{stk}{\contCons{\actUnbind{x}}{k}}
        }{
            \running{\theta}{\phi}{stk}{k}
        }
    }

    \infer[ST-MatchConstr]{
        \theta(x) \mapsto t
    }{
        \stateTrans{
            \running{\theta}{\phi}{stk}{\contCons{\actBackMatch{x}{p}}{k}}
        }{
            \running{\theta}{\phi}{stk}{\contCons{\actDoMatch{p}{t}}{k}}
        }
    }
    \end{mathpar}
   \caption{Algorithmic Semantics Part 1}
   \label{fig:all-algsem-1}
\end{figure*}

\begin{figure*}
    \begin{mathpar}
     \infer[ST-Match-Exists]{
        k' = \actUnbind{x} \texttt{::} k
    }{
        \stateTrans{
            \running{\theta}{\phi}{stk}{\contCons{\actDoMatch{\exists x.p}{t}}{k}}
        }{
            \runningNOFUN{\theta}{stk}{\contCons{\actDoMatch{p}{t}}{k'}}
        }
    }
   
    \infer[ST-Match-MatchConstr]{
        k' = \actBackMatch{x}{p'} \texttt{::} k
    }{
        \stateTrans{
            \running{\theta}{\phi}{stk}{\contCons{\actDoMatch{\backMatch{p}{x}{p'}}{t}}{k}}
        }{
            \running{\theta}{\phi}{stk}{\contCons{\actDoMatch{p}{t}}{k'}}
        }
    }

     \inferrule[ST-Match-Fun-Var-Bind]{
         \neg\exists f'. \phi(x) \mapsto f'\\
         k' = \left[\actDoMatch{p_1}{t_1},\dots,\actDoMatch{p_n}{t_n}\right]
     }{
         \stateTrans{
             \running{\theta}{\phi}{stk}{\contCons{\actDoMatch{F(p_1,\dots,p_n)}{f(t_1,\dots,t_n)}}{k}}
         }{
             \running{\theta}{\phi \cup \{F \mapsto f\}}{stk}{k' \texttt{++} k}
         }
     }\\

     \inferrule[ST-Match-Fun-Var-Bound]{
         \phi(x) \mapsto f
     }{
         \stateTrans{
             \running{\theta}{\phi}{stk}{\contCons{\actDoMatch{F(p_1,\dots,p_n)}{f(t_1,\dots,t_n)}}{k}}
         }{
             \running{\theta}{\phi}{stk}{k' \texttt{++} k}
         }
     }

     \inferrule[ST-Match-Fun-Var-Conflict]{
         \left(\phi(x) \mapsto g \wedge f \neq g\right) \vee m \neq n
     }{
         \stateTrans{
             \running{\theta}{\phi}{stk}{\contCons{\actDoMatch{F(p_1,\dots,p_n)}{f(t_1,\dots,t_m)}}{k}}
         }{
             \backtrack{stk}
         }
     }

     \infer[ST-Match-Mu]{
        p' = p[\mu P(x_1,\dots,x_n)/P][y_i/x_i]
    }{
        \stateTrans{
            \running{\theta}{\phi}{stk}{\contCons{\actDoMatch{\mu P(x_1,\dots,x_n)[y_1,\dots,y_n].p}{t}}{k}}
        }{
            \running{\theta}{\phi}{stk}{\contCons{\actDoMatch{p'}{t}}{k}}
        }
    }

    \end{mathpar}
   \caption{Algorithmic Semantics Part 2}
   \label{fig:all-algsem-2}
\end{figure*}
\fi

\end{document}
\endinput